\documentclass[letterpaper, 10pt, conference]{ieeeconf}      % Use this line for a4 paper

\IEEEoverridecommandlockouts                              % This command is only needed if 
\usepackage{graphics} % for pdf, bitmapped graphics files
\usepackage{epsfig} % for postscript graphics files
\usepackage{mathptmx} % assumes new font selection scheme installed
\usepackage{times} % assumes new font selection scheme installed
\usepackage{amsmath} % assumes amsmath packet installed
\usepackage{amssymb}  % assumes amsmath packet installed
\usepackage{bm}
\usepackage[hidelinks]{hyperref}
\let\labelindent\relax %to avoid error in enumitem packet
\usepackage{enumitem} %to add letters for enumerations (eg A1, A2)

\usepackage[style=ieee]{biblatex} %Imports biblatex packet
\AtEveryBibitem{\clearfield{url}}
\AtEveryBibitem{\clearfield{doi}}
\AtEveryBibitem{\clearfield{isbn}}
\AtEveryBibitem{\clearfield{issn}}
\AtBeginBibliography{\small}

\usepackage{amsthm} 

\newtheoremstyle{boldstyle}%
{}% measure of space to leave above the theorem. E.g.: 3pt
{}% measure of space to leave below the theorem. E.g.: 3pt
{\itshape}% name of font to use in the body of the theorem
{}% measure of space to indent
{\bfseries}% name of head font
{.}% punctuation between head and body
{.5em}% space after theorem head; " " = normal interword space
{\thmname{#1}\thmnumber{ #2}\thmnote{ (#3)}}
\theoremstyle{boldstyle}
\newtheorem{theorem}{Theorem}[section]

\newtheorem{lemma}[theorem]{Lemma}

\newtheorem{definition}{Definition}

\graphicspath{ {./Pics/} }

\usepackage{tikz}
\usepackage{pgfplots}
\pgfplotsset{compat=1.18}
\usetikzlibrary{arrows,shapes,backgrounds,patterns,fadings,decorations,
	positioning,external,arrows.meta,pgfplots.fillbetween, pgfplots.groupplots}

\newcommand\copyrighttext{%
	\footnotesize \textcopyright 2024 IEEE. Personal use of this material is permitted.  Permission from IEEE must be obtained for all other uses, in any current or future media, including reprinting/republishing this material for advertising or promotional purposes, creating new collective works, for resale or redistribution to servers or lists, or reuse of any copyrighted component of this work in other works.}
\newcommand\copyrightnotice{%
	\begin{tikzpicture}[remember picture,overlay]
		\node[anchor=south,yshift=5pt] at (current page.south) {\fbox{\parbox{\dimexpr\textwidth-\fboxsep-\fboxrule\relax}{\copyrighttext}}};
	\end{tikzpicture}%
}

\title{\LARGE \bf
A Robust Model Predictive Control Method for Networked Control Systems
}

\author{Severin Beger \and Sandra Hirche% <-this % stops a space
\thanks{S. Beger and S. Hirche are with the Chair of Information-oriented 	Control, TUM School of Computation, Information, and Technology, Technical University of Munich, 80333, Munich, Germany. Email: \texttt{\{severin.beger, hirche\}@tum.de}}%
\thanks{This work was supported by the Federal Ministry of Education and 	Research of Germany (BMBF) as part of the project 6G-life (project identification number 16KISK002).}% <-this % stops a space
}
\begin{document}

\maketitle
\copyrightnotice
\thispagestyle{empty}
\pagestyle{empty}

%%%%%%%%%%%%%%%%%%%%%%%%%%%%%%%%%%%%%%%%%%%%%%%%%%%%%%%%%%%%%%%%%%%%%%%%%%%%%%%%
\begin{abstract}

 Robustly compensating network constraints such as delays and packet dropouts in networked control systems is crucial for remotely controlling dynamical systems. This work proposes a novel prediction consistent method to cope with delays and packet losses as encountered in UDP-type communication systems. The augmented control system preserves all properties of the original model predictive control method under the network constraints. Furthermore, we propose to use linear tube MPC with the novel method and show that the system converges robustly to the origin under mild conditions. We illustrate this with simulation examples of a cart pole and a continuous stirred tank reactor.

\end{abstract}

%%%%%%%%%%%%%%%%%%%%%%%%%%%%%%%%%%%%%%%%%%%%%%%%%%%%%%%%%%%%%%%%%%%%%%%%%%%%%%%%
\section{Introduction}
The advancement in data driven control requires solutions for robust networked control systems (NCS) in order to outsource heavy computation. This enables the usage of complex, data hungry methods for small embedded systems, while providing high flexibility for system design and scalability as well as easy maintenance.
However, the use of communication networks introduces additional challenges for closed loop control, such as time delays and packet losses. Extensive research has been conducted to develop stability analysis and controller design tools to cope with these flaws (cf. \cite{J.P.Hespanha.2007}), mostly in the setting of delays smaller than a sampling step. \\
For applications with high sample rates, such as robotic systems, networked communication results in much longer delays of several sample time steps. Model Predictive Control (MPC) is suited well for handling these network-related challenges \cite{YangShi.2021}, as its predictive nature can be explicitly used to compensate for delays and dropouts.
One of the first to suggest using MPC for delay compensation was \cite{A.Bemporad.} for a teleoperation scenario. 
In \cite{L.Grune.2009} the authors introduced a core property for deterministic predictive networked control methods under delays and packet losses, which they called \emph{prediction consistency}. When extending a nominal predictive control method to be prediction consistent under delays and dropouts, properties from the nominal closed loop system, such as stability, are conserved.
Several works have proposed prediction consistent methods to apply MPC in NCS, such as \cite{Pin.2009b} and \cite{Varutti.2009a}. Some extensions concern using packet-based communication networks without acknowledgments (also known as UDP-like)  \cite{Pin.2010}, event-based methods \cite{Varutti.2011}, and input to state stability \cite{Findeisen.2011,G.Pin.2011}. In \cite{Varutti.2009b} the authors address the difficulty of time synchronization between components in a network. 
More recent works aims to simplify the compensation schemes such as \cite{G.Pin.2021} and focus on realistic scenarios, e.g. using a WiFi network\cite{Pezzutto.2022}.
So far, little attention has been given towards using robust MPC methods in the networked setting. However, many remotely controlled systems have some limited computing power and access to the most recent state measurements, which can be exploited efficiently. 
We address this idea in the paper at hand. \\
Our contribution is twofold. Firstly, we propose a novel method for ensuring prediction consistency when using a predictive control method over a communication network subject to bounded time-varying delays and packet dropouts. It relies on UDP-like communication and does not require time synchronization, thus making it suitable for general-purpose communication networks. Secondly, we investigate how tube MPC can be used in this setting to robustify a process subject to bounded additive noise. Necessary limitations on the tube due to the network effects are derived. \\
The work is structured as follows: at first, we state the considered dynamics under constraints and our assumptions on the communication network. Subsequently, we present our method for ensuring prediction consistency. Then, we derive our main results on prediction consistency of our method and bounds on the tube due to network influences. Two simulation studies are presented to substantiate the efficacy of the proposed method. At last, we discuss our results and conclude our work.

\section{Problem Statement}

Consider the disturbed linear discrete-time system subject to additive noise $w$
\begin{align}
    x[k+1]=Ax[k]+Bu[k]+w[k], \quad x[0] = x_0
    \label{eq:dynamics}
\end{align}
where the states, inputs, and disturbances are bounded in the sets
\begin{align}
    \label{eq:constraints}
    x[k] \in \mathbb{X}, \,\,  u[k] \in \mathbb{U}, \,  w[k] \in \mathbb{W} \qquad \forall k \in \mathbb{N}_{0}. 
\end{align} 
The sets $\mathbb{U}\subset \mathbb{R}^m$ and $\mathbb{W}\subseteq \mathbb{R}^n$ are compact and convex polytopes, while $\mathbb{X}\subset \mathbb{R}^n$ is a convex, bounded polyhedron.\\
We assume sensors and actuators to be collocated. The control loop is closed with a remote controller over a non-acknowledged UDP-like packet-based communication network, which imposes communication constraints in the form of time-varying delays and packet losses. \\
Additionally, we assume that the logic unit at the plant side, in the following denoted as a local controller, has enough computational power to handle our proposed algorithm and to execute a linear state feedback law. \\
We consider a remote controller in the network that possesses close to unlimited computational resources, runs the predictive control algorithm, and takes delays and packet losses into account. Both sides have some memory capabilities in the form of buffers. The system setup is depicted in fig. \ref{fig:SystemOverview}. 
\begin{figure}[t]
\centering
\includegraphics[width=.5\textwidth]{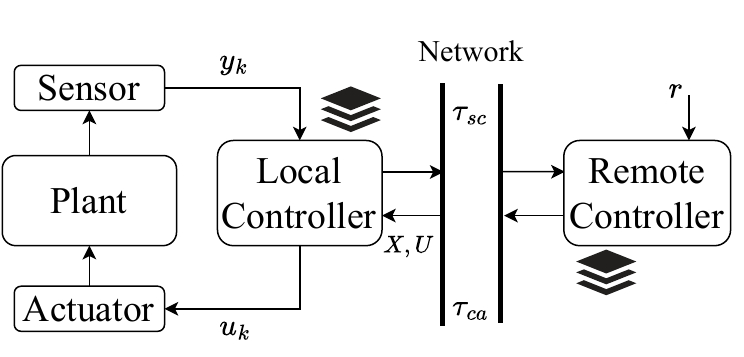}
\caption{Considered setup. The local controller forwards measurements to the remote side and receives actuation signals from the remote controller over a lossy network.}
\label{fig:SystemOverview}
\end{figure}\\
Consider the following assumptions on the computing units and network behavior: \\
\begin{enumerate}[label=\textbf{A\arabic*)},ref=A\arabic*]
    \item \label{A1}Every message sent from sensor to controller is time-stamped with the time  $t_p$ when the measurement was taken.
    \item \label{A2}The sensor-to-controller delay $\tau_{sc}$ as well as the controller-to-actuator delay $\tau_{ca}$ are upper bounded\footnote{We treat all delays longer than their respective upper bound as packet losses.}: $0 \leq \tau_{sc}[k] \leq \bar{\tau}_{sc}$, $0 \leq \tau_{ca}[k]\leq \bar{\tau}_{ca}$ .  
    \item \label{A3}The upper bound of the roundtrip time $\bar{\tau}_{RTT} = \bar{\tau}_{sc}+ \bar{\tau}_{ca}$ for a successful transmission without dropouts from sensor-to-controller-to-sensor is known, where $\bar{\tau}_{ca}$ includes computing time.
    \item \label{A4}The number of consecutive packet losses over the complete loop (from sensor to actuator) is upper bounded by $\bar{n}_l \in \mathbb{N}_0$.
    \item \label{A5}The local controller at the plant has enough computational power to execute a linear feedback policy.
\end{enumerate}
\vspace{0.3cm}
Note that we do not require time synchronization. 
Our scheme is designed to rely exclusively on timestamps from the plant side. \\
In this work, we consider a robust MPC strategy for linear systems, known as tube MPC, as described e.g. in \cite{Mayne.2005}.
Consider the nominal system without disturbances, $w = 0$, as
\begin{align}
    z[k+1] = Az[k]+Bv[k].
    \label{eq:nominaldynamics}
\end{align}
The following optimal control problem (OCP) $\mathbb{P}(z)$ is solved at the remote controller:\\
\begin{align}
    \label{eq:OCP}
    \min_{\bm{V}}\, \bigl(z[N]^T&Pz[N] + \sum_{j=0}^{N-1} z[j]^TQz[j]+v[j]^TRv[j]\bigr)\\
    \text{s.t.} \quad&z[j+1]=Az[j]+Bv[j] \nonumber\\
     \quad &z[j]\in \overline{\mathbb{X}}, v[j] \in \overline{\mathbb{U}}, z[N]\in\overline{\mathbb{X}}_{f},\nonumber
\end{align}
where $\bm{V}[k] = \{v[0|k],\cdots,v[N-1|k]\}$ is the resulting input sequence with the notation $v[j|k]$ denoting the input at time $j+k$ computed at time $k$ based on a measurement $x[k]$. $i \in [0,N-1]$, $N$ is the considered horizon.\\
Additionally, $P=P^T\succeq 0$, $Q=Q^T\succeq 0$, and $R=R^T\succ 0$. The matrices, as well as the terminal set $\overline{\mathbb{X}}_{f}$, must be chosen carefully to provide recursive feasibility and stability (see e.g. \cite{Rawlings.2017}).\\
The sets $\overline{\mathbb{X}} = \mathbb{X} \ominus \mathbb{S}$, $\overline{\mathbb{U}} = \mathbb{U} \ominus K\mathbb{S}$ and $\overline{\mathbb{X}}_{f} = \mathbb{X}_f \ominus \mathbb{S}$ are the state and input constraints tightened with a robust positive invariant tube $\mathbb{S}$.
A precalculated static feedback gain $K$ ensures that the system under disturbance converges to the nominal trajectory. The considered feedback law at the real system is
\begin{align}
    u[k] = v^*[0|k] + K (x[k]-z[k]),
    \label{eq:tubeMPClaw}
\end{align}
where $v^*[0|k]$ is the first optimal input, computed from $\mathbb{P}(z)$ with $z_0 = x[k]$, whereas the second term provides a converging behavior towards the predicted state sequence.\\
Because we know the roundtrip time $\tau_{RTT}$ as well as the time of a measurement $t_p$, we can use our system model to predict the state for the next predicted time of arrival and solve the OCP (\ref{eq:OCP}) to calculate inputs for that time $t_{p,d}$. Through this we counter act communication delays. By sending the complete computed input sequence, the local controller may use the later values of the sequence as a backup in case no new control values are delivered due to a packet dropout.\\
The main problem arises from predicting the future state for solving $\mathbb{P}[\hat{x}[k]]$ consistently without considering past predicted inputs that did not arrive at the plant side. The following definition formalizes this thought, where $t_a$ denotes the time of application at the actuator of a control input sequence $\bm{V}[t_a]$, while $t_s$ is the time of sensing of a measurement:
\begin{definition}\cite{L.Grune.2009} \label{def:PredictionConsistency}
    \begin{enumerate}[label=(\roman*)]
        \item We call a feedback control sequence $\bm{V}(t_{a})$ consistently predicted if the control sequence $[\hat{u}(t_s),\ldots,\hat{u}(t_{a}-1)]$ used for the prediction of $\hat{x}[t_{a}]$ equals the control sequence $[u(t_s),\ldots,u(t_{a}-1)]$ applied at the actuator.
        \item We call a networked control scheme prediction consistent if at each time $k$ the computation of $u[k]$ according to \eqref{eq:tubeMPClaw} in the actuator is well defined, i.e., $k-t_{a} \leq N-1$
        and $\bm{V}[t_a]$ is consistently predicted.
    \end{enumerate}
\end{definition}
\section{A simple consistent prediction method}
In this section, we define the behaviors of the local controller at the plant side and of the remote controller in the network, such that they result in a prediction consistent control scheme.  \\
In the following, we will denote the time at the plant and at the remote controller as $t_p$ and $t_c$, respectively. 
All considered time stamps, delays and dropout are a multiple of the sampling time $T$ and thus $t_j,\tau_j, n_j \in \mathbb{N}_{0} \, \forall j$. Furthermore, all input trajectories are tagged with an ID $i \in \mathbb{N}$. We denote the buffer at the plant side as $B_p$ and at the remote controller as $B_c$. The packets sent over the network are labeled similarly as $P_p$ and $P_c$. \\
Initially, the system needs to be in a save state, e.g. an equilibrium. Our method is initiated with the first input sequence that reaches the plant side, based on the measurement $x_0$ at time $t_p = 0$.\\
First, we develop the behavior of the local control unit at the plant side. It is sample-based and executes its algorithm with a sample rate $T_s$. At the beginning of each cycle, the system needs to check for a new control packet from the remote controller. Let us denote this operation with the Boolean variable $m_{new} \in \{0,1\}$. Each control packet $P_c$ includes an ID $i$, the corresponding desired time of application at the local controller $t_{p,d}(i)$, the current predicted trajectory of states $\bm{X}(i)= \{x[0|t_{p,d}(i)],x[1|t_{p,d}(i)],\cdots,x[N|t_{p,d}(i)]\}$ and the associated trajectory of inputs $\bm{V}(i) = \{v[0|t_{p,d}(i)],v[1|t_{p,d}(i)],\cdots,v[N-1|t_{p,d}(i)]\}$. Additionally, it contains the ID $i_{c,last}$ of the last input trajectory, which was used to predict $\hat{x}[\bar{\tau}_{RTT}|x[t_p(i)]]$. The local controller needs the latter to determine a consistently predicted input. A control packet can be denoted as $P_c = [i, t_{p,d}(i), \bm{X}(i), \bm{V}(i), i_{c,last}]$.
The local controller now checks, whether the time of execution is smaller or equal to the current time at the plant side: $t_{p,d} \leq t_p$. If so, the contents of the new control packet are put into the buffer $B_p$. If not, the new packet is discarded. \\
Then, the controller checks the buffer for an input expected for the current time step: $\exists i  \in B_p \text{ for which } t_{p,d}(i) = t_p.$ If such an input exists, we need to check if it is prediction consistent, by comparing 
\begin{align} \label{eq:PredictionConsistencyPlant}
    i_{c,last} = i_{p,last}.
\end{align} The trajectory is used, if this holds true. We update the predicted input $v^*[t_p] = v[0|t_{p,d}(i)]$, the predicted state $\hat{x}[t_p] = x[0|t_{p,d}(i)]$ and the ID of the last applied input $i_{p,last} = i$. The internal counter of successively used inputs of a single input trajectory is (re)set to $c_p = 1$. Older control packets with $t_{p,d}(i) < t_p$ can be pruned from the buffer. \\
In case the input was not prediction consistent, we delete it and all trajectories for later times $t_{p,d}(i) > t_p$ from the buffer. Now, we reuse the last consistent trajectory with ID $i_{p,last}$ and choose $v^*[t_p] = v[c_p|t_{p,d}(i_{p,last})]$ as well as the corresponding predicted state $\hat{x}[t_p] = x[c_p|t_{p,d}(i_{p,last})]$. Then we increment the counter of used inputs by one.\\
If no new control packet arrives at the plant, the procedure is the same as for the situation without prediction consistency, only that the buffer is not pruned. \\ 
In any case, we measure our system state $x[t_p]$ and use it to compute the input for the plant
\begin{align}
    u[t_p] = v^*[t_p] + K(x[t_p] - \hat{x}[t_p])
\end{align}
and applies it.
Next, the plant side sends a packet $P_p$ to the controller, which contains the current measurement $x[t_p]$, the associated time of measurement $t_p$ and the ID of the last considered input trajectory $i_{p,last}$ resulting in $P_p = [x[t_p],t_p,i_{p,last}]$.
A state flow diagram of the plant side algorithm is given in Fig. \ref{fig:StateFlow_Plant}.\\
\begin{figure}[!ht]
\centering
\includegraphics[width=.35\textwidth]{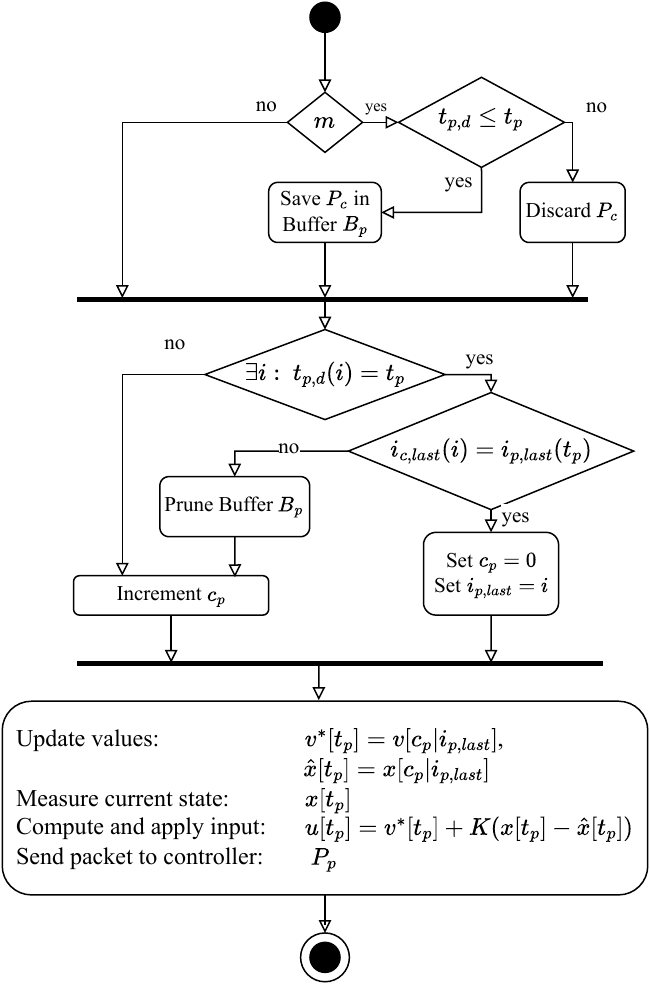}
\caption{Plant side algorithm. Top part handles incoming control packets, middle part checks on prediction consisteny, bottom task computes and executes the control input.}
\label{fig:StateFlow_Plant}
\end{figure}
Next, we turn our attention towards the remote controller. It runs in two different states, which are differentiated with the Boolean variable $s \in \{0,1\}$. The first is the nominal operation for which $s=0$. The other is the recovery mode, which is for situations where prediction consistency does not hold. Similar to the plant side, we use a Boolean variable $m_{new} \in \{0,1\}$ to determine the situations when a new message from the plant has arrived. This leaves us with four distinct situations at the remote controller. \\
The simplest case occurs when no new measurement arrives, and the controller is in nominal mode: $(m,s) = (0,0)$. We consider a reactive controller, which only computes a new input when it receives new data. Thus, the controller stays idle when no new packet arrives from the plant. \\
If a measurement arrives, regardless of the controller state, we first need to check whether the applied input at the time of the measurement matches the expected input for that point in time:
\begin{align} \label{eq:PredictionConsistentIDs}
    i_{c}(t_p) = i_{p,last}.
\end{align}
Here, $i_c(t_p)$ describes the ID of the input trajectory in $B_c$ with the smallest difference of the desired application time to the measurement time. 
\begin{align}
    i_c(t_p) &= \min_i (t_p-t_{p,d}(i)) ~~~~\text{s.t. } t_{p,d}(i) \leq t_p
\end{align}
If (\ref{eq:PredictionConsistentIDs}) holds true, the corresponding input was prediction consistent at $t_p$. If the system was in recovery mode, it is necessary to check, whether the newly arrived trajectory was a correction trajectory for the existing error. To be able to do this check, we assign every correction trajectory a marker $e(i) = i_{p,last}$, which carries the ID of the last applied input trajectory, when the prediction inconsistency occured. If the error correction check holds, the error is now assumed to be resolved, and hence we set $s=0$ to return to the nominal mode. In this case, we need to prune our buffer, $B_c$, as there are still some correction trajectories saved, which were certainly not used at the plant. Thus, we prune all buffer values, for whose IDs the condition 
\begin{align} \label{eq:PruneBufferOnCorrection}
    e(i) =e(t_{p,last}) \wedge i \neq i_{p,last}  ~ \forall i \in B_c
\end{align}
holds. Now $B_c$ and $B_p$ are consistent again and we can continue with the state prediction.
\\
The controller uses the new measurement $x[t_p]$ to predict the next state at $t_{p,d}(i) = t_p + \bar{\tau}_{RTT}$ by a simple forward role out of the dynamics with the known model and the expected last inputs for the timesteps between $t_p$ and $t_{p,d}(i)$. In nominal mode, the controller is optimistic and assumes that all sent packets since $t_p$ have arrived. From the state prediction $\hat{x}[t_{p,d}]$ it solves the OCP (\ref{eq:OCP}).
A new ID $i$ is computed by increasing a counter and saved in the buffer alongside the calculated inputs, states, time of application at the plant side, and the ID of the last used input trajectory $ I_{c,last} = i-1$ and simultaneously sent to the plant. \\
If (\ref{eq:PredictionConsistentIDs}) is not true, an error occurs, and the remote controller switches to or stays in recovery mode $s=1$. We know now that all input predictions for times later than the last applied input at $t_p$ are inconsistent. Therefore, the plant will reject them. Thus, the controller can prune its buffer $B_c$ and delete all predictions $i$  with 
\begin{align} \label{eq:PruningOnError}
    t_{p,d}(i) > t_p ~ \forall i \in B_c.
\end{align} 
Now, the buffers at the controller and plant side are equal. Next, a new input sequence, which we call correction trajectory, needs to be computed. When in recovery mode, the controller is pessimistic. It assumes no inputs have arrived since the last confirmed input with $i = i_{p,last}, ~ i \in B_c$. It uses the input trajectory of that old prediction together with the known nominal system model to calculate the new state at the time $t_{p,d} = t_p + (t_p - t_{p,d}(i_{p,last})) + \bar{\tau}_{RTT}$. As before, all necessary values are saved in the local buffer, this time with an error code $e(i)$, and put into a packet to be sent to the plant. What is important is that the ID of the last used input is set correctly as $i_{c,last} = i_{p,last}([t_p])$. This ensures that if some recovery packets from the controller to the plant are lost due to the network, the next packet that arrives will certainly be prediction consistent.\\ 
Finally, we look at the case where no new measurement arrives at the controller, and it is in recovery mode $(m,s) = (0,1)$. We use the last consistently predicted input trajectory and the last available measurement to predict a new state and calculate new inputs. Note that this is the only situation in which the controller calculates and sends new values, even though it did not receive new measurements. The state flow diagram for the logic used at the remote controller side is given in Fig. \ref{fig:StateFlow_RM}. 

\begin{figure}[ht]
\centering
\includegraphics[width=.35\textwidth]{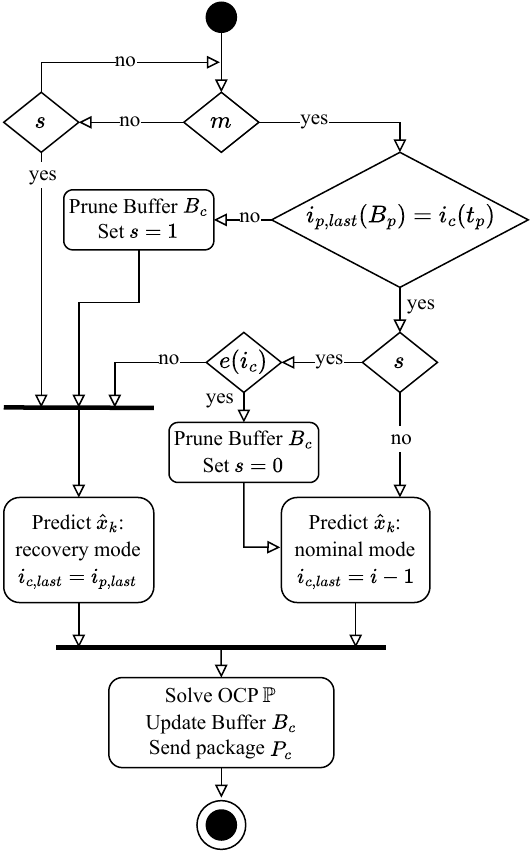}
\caption{Remote controller algorithm. Top part describes checks for measurement and prediction consistency, the middle part determines the state prediction, the bottom part solves the OCP.}
\label{fig:StateFlow_RM}
\end{figure}

\section{Properties of the method}

In \cite[Theorem 2.2]{L.Grune.2009} it is stated that a system remains practically asymptotic stable for a predictive controller, given that the nominal closed loop system is practically asymptotic stable and the used prediction method is prediction consistent. Along this lines it is noted that the prediction behavior can be separated from the robustness of the scheme.\\
Therefore, the analysis of our method is twofold: first, we examine all cases of occurring delays and packet losses to show the prediction consistency of our method. Secondly, we analyze the error propagation in the worst-case scenario and derive a bound for the robustly invariant tube $\mathbb{S}$ to guarantee the robust constraint satisfaction of our closed loop system in the networked setting. Then, when applying the findings jointly, we can extend the result to a robustly asymptotic stable system. \\
\begin{lemma}\label{thm:Lemma1}
    If assumptions (\ref{A1}-\ref{A4}) hold, the proposed method for delay and packet loss compensation using model predictive control methodologies in networked control systems under the influence of delays and packet dropouts is prediction consistent in the sense of Definition \ref{def:PredictionConsistency}.
\end{lemma}
\begin{proof}
    The proof can be found in the appendix.
\end{proof}

In the following we show the necessary lower bound for the disturbance invariant set, such that the system is robustly asymptotically stable with the proposed method.
\begin{theorem}\label{thm:Theorem2}
    Given a linear system with additive bounded disturbance as in (\ref{eq:dynamics}) with constraints (\ref{eq:constraints}) and a robustly asymptotically stabilizing tube MPC law as in (\ref{eq:tubeMPClaw}) based on the OCP (\ref{eq:OCP}) and a precomputed, stabilizing Feedback gain $K$. Assume that assumptions (\ref{A1}-\ref{A5}) hold and the initial state $x_0$ lies in the feasible set. Assume that the chosen time horizon fulfills (\ref{eq:MaximumHorizon}) and the tube is lower bounded by  \begin{align} 
         \bigoplus_{j=0}^{L}(A+BK)^{j-1}\mathbb{W}\subseteq\mathbb{S}
        \label{eq:Tube}
    \end{align}
    where $L = \max[N,2 \bar{n}_l + 3\bar{\tau}_{RTT}-1]$. Then the system is robustly asymptotically stable.
\end{theorem}
\begin{figure}[b]
\centering
\includegraphics[width=.45\textwidth]{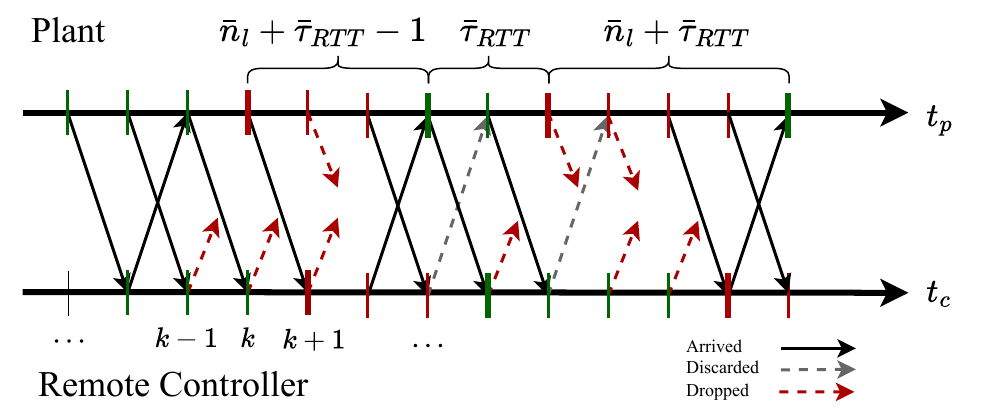}
\caption{Worst case delay and dropout scenario of the presented method. Arrows represent packets. $\bar{\tau}_{RTT} = \bar{n}_{l} = 2$.}
\label{fig:WorstCaseTube}
\end{figure}
\begin{proof}
    From Lemma \ref{thm:Lemma1} the prediction consistency of our scheme follows. By making use of  \cite[Theorem 2.2]{L.Grune.2009} we can assume that an already nominally asymptotically stable MPC method preserves this property, if a prediction consistent method is used. \\ 
    Robust asymptotically stable behavior for the nominal closed loop system is provided by linear tube MPC as long as our system operates within disturbance invariant sets. %It is desirable to find the minimum disturbance invariant set to maximize the set of feasible initial states and ensure higher performance through less conservative tightening of the input constraints. 
    To derive a bound on the tube in our scenario, we need to consider the worst possible error from the considered disturbances at the maximum prediction step $N$ for a given input sequence.\\
    Given a state measurement $x[0]$ at time $k=0$, and the current error due to disturbance $e[0] = 0$, the error develops as follows
    \begin{align}\label{eq:errordevelopment}
        e[k] &= (A+BK) e[k-1] + w[k-1]\\  \nonumber
        &= \sum_{j=1}^{k}(A+BK)^{j-1} w(k-j).
    \end{align}
      Let's consider a drop from controller to actuator. As a result, no new input trajectory arrives at the plant, and it has to reuse the previous one. Let's assume the remote controller is notified about the prediction inconsistency with the measurement $t_p=k$. It starts to compute and send correction trajectories based on this measurement in every time step, until the error is resolved. In the worst case for disturbance, all measurements are dropped for the maximum amount of steps $\bar{n}_l$, but right before a new measurement is received, the first successful transmission of a correction trajectory based on the measurement at $k$ takes place after $\bar{n}_l -1$ steps. It is applied at $t_{p,cor} = k + \bar{n}_l + \bar{\tau}_{RTT}-1$. At this point in time the predicted state and the actual state differ by $e[t_{p,cor}] = x[t_{p,corr}] - \hat{x}[\bar{n}_l + \bar{\tau}_{RTT}-1 | k]$. Now, in the worst-case situation for the switch from recovery mode to nominal mode, the input trajectory, which is computed from the first measurement after the arrival of the correction trajectory, drops out. Again, recovery mode needs to be triggered, and a correction trajectory needs to be sent. Assuming a complete blackout for $\bar{n}_l$ steps followed by the maximum roundtrip time to deliver the correction trajectory, the time between the two correction trajectories arriving at the plant is $2\bar{\tau}_{RTT} + \bar{n}_l$. At this point in time, the error originating from the first error at time $k$ has progressed to its maximum value $e[2\bar{\tau}_{RTT} + 3 \bar{n}_l -1]$.  Thus, the maximum of either
      \begin{align} \label{eq:MaxDistTime}
          L = 3\bar{\tau}_{RTT} + 2 \bar{n}_l -1
      \end{align} or the freely chosen prediction horizon $N$ gives us the number of steps to consider for error development and, therefore, a bound on the minimum disturbance invariant set. This concludes the proof.
\end{proof}
Fig. \ref{fig:WorstCaseTube} illustrates the described worst-case situation with $\bar{n}_l = \bar{\tau}_{RTT} = 2$.

\section{Simulation Examples}

In the following, we present two simulation examples. For both we use a simulated network composed of two Markov chains, one for a delay state and the other for dropouts. While the latter uses two states (dropout or successful transmission), the former has three states, representing the load of the network. Each state operates on a different Weibull distribution, to represent situations of low, medium and high network traffic.

\subsection{Cart Pole}

The first use case is a cart pole system. We use the linearized version of a real system introduced in \cite{Branz.2022} and used in \cite{Pezzutto.2022} amongst others. The state vector $x = [s,\alpha, \dot{x},\dot{\alpha}]^T$ contains the position $s$ of the device on a straight line
as well as its tilt angle $\alpha$. We use a sampling time of $T_s = 0.01s$. The corresponding discretized dynamic matrices are 
\begin{align}
A = \begin{bmatrix}
    1 & 0.002 & 0.010 & 0 \\
    0& 1.003 & 0 & 0.010\\
    0 & 0.437 &0.963 & 0.0353 \\
    0 & 0.551 & 0.019 & 0.981
    \end{bmatrix}, & ~B = \begin{bmatrix}
    0\\
    0\\
    0.048\\
    0.032
\end{bmatrix}.
\end{align}
For the sake of simplicity we assume full state measurements. We set $N=50$, $R=1$, and $Q = diag(10,1000,1,1)$. Furthermore, we obtain $K_{LQR}$ as well as $Q_N=P$ from solving the discrete algebraic Riccati equation with the proposed matrices. Finally, for our network, we set $\bar{\tau}_{RTT} = 7$ and $\bar{n}_l=3$. \\
The considered delays are $w[k] = \begin{bmatrix}0&0&10&1\end{bmatrix}^Td[k] ~~ \text{with } d[k] \in [-1,1]$.
Two constraints are imposed on the system. Firstly, the tilting angle needs to be in $[-0.2, 0.2]$rad to ensure that the linearization holds. Additionally, the input needs to be between $[-20, 20]$V, which is a physical limitation of the real system. \\
%Plot Position Balance Rob
\begin{figure}[t]
\centering
\includegraphics[width=.5\textwidth]{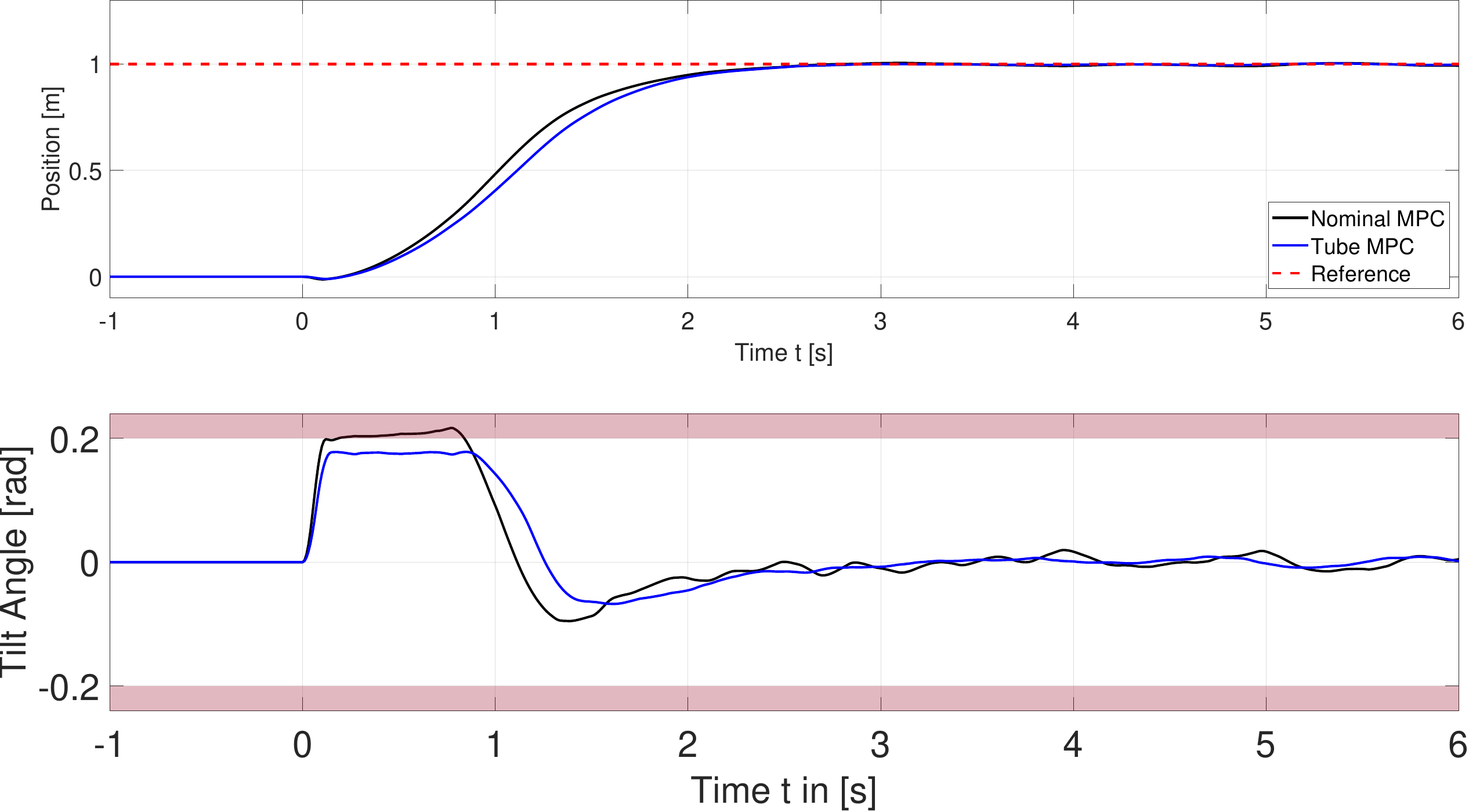}
\caption{Position and tilt angle of the balance robot. We compare a nominal MPC with the tube based MPC method.}
\label{fig:PositionBalanceRob}
\end{figure}
Fig. \ref{fig:PositionBalanceRob} shows position and tilt angle over time. We compare a nominal linear MPC with the tube MPC approach, both applying the prediction consistent method. As our horizon is much larger than the assumed delays and dropouts, the MPC is able to stabilize the system in both scenarios despite the imposed communication constraints. However, as a reaction to disturbance may be delayed up to $70\text{ms}$, the nominal MPC fails to stay within the tilt angle constraints. The tube MPC approach, on the other hand, fulfills all constraints, although the considered tightening of the bounds seems to be overly conservative. In Fig \ref{fig:RTTBalanceRob} the roundtrip time is depicted. The red line shows the expected maximum delay $\bar{\tau}_{RTT}$. Despite the frequent packet losses the method manages to safely guide the system to its desired location, even under substantial disturbances.

\subsection{Continuous Stirred Tank Reactor}
Our second use case is from the process industry. It is interesting, as it is nonlinear as well as operates on a much different time scale than the cart pole example. We consider an adiabatic continuous stirred tank reactor (CSTR) with the following dynamics
\begin{align}
\frac{dT}{dt} &= \frac{F}{V}(T_{A0}-T) + \sum_{i=1}^3\frac{-\Delta H_i}{\rho c_{p}} k_{i0}e^{\left(\frac{-E_{i}}{RT}\right)}C_{A} + \frac{Q}{\sigma c_{p} V_r} \\
\frac{dC_A}{dt} &= \frac{F}{V_r}(C_{A0}+\Delta C_{A0}-C_{A})- \sum_{i=1}^3k_{i0}e^{\left( \frac{-E_i}{RT} \right)}C_{A} 
\end{align}
with states $x = [T,C_A]^T$, which denote the substance temperature as well as the molar concentration of the considered reactant. As input we apply or remove heat through $u = Q$, which is bounded by $\vert Q \vert \leq 10^5 \text{kJ/h}$. Furthermore, we consider an unknown, bounded time-varying uncertainty $\Delta C_{A0} \in [-0.5,0.5] \text{mol/l}$. All parameters of the system can be found in \cite{Liu.2008}. For our simulation we consider $T_s = 0.025\text{h}$, $T_{end}=0.6h$, $N = 10$, $Q = diag(1, 1000)$, $R = 10^{-6}$. The cost function is quadratic, as shown in the OCP (\ref{eq:OCP}). However, we do not consider terminal costs but imply the final state $x_N$ to lie within the ellipse $(x_1-x_{1,d}) ^2 + \frac{(x_2-x_{2,d})^2}{0.2^2} \leq 1$, meaning that the final state shall not deviate more than $1\text{K}$ and $0.2\text{mol/l}$ from the desired state, which is $x_d = [388\text{K},3.59 \text{mol/l}]$. For the netwok we choose $\bar{\tau}_{RTT} = 4$ and $\bar{n}_l = 2$ with a prediction horizon $N=10 =2\bar{\tau}_{RTT} + \bar{n}_{l} $. As our system is nonlinear, we use an adaptive approach linearizing around every point of a predicted MPC trajectory and computing the corresponding LQR gains. These gains are then forwarded alongside the inputs and state trajectories to the local controller and iterated like the predicted input. The states over time are plotted in fig. \ref{fig:PositionsCSTR}, whereas fig. \ref{fig:PhasePortraitCSTR} shows the phase portrait of the system. As can be seen, the nominal MPC cannot stay within the desired terminal set and seems to become unstable towards the end of the simulation. The tube MPC, on the other hand, manages to keep the system in the desired terminal set, even though the disturbance on the molar concentration influences the system heavily. \\
It must be noted that we don't use the proposed constraint tightening in this scenario and thus don't guarantee save constraint fulfillment anymore. Nonetheless, this approach seems to be an effective strategy, especially for systems like the CSTR, where time and input constraints are not of too much concern.\\
%%% Plot RTT Cart Pole
\begin{figure}[t]
\centering
\includegraphics[width=.5\textwidth]{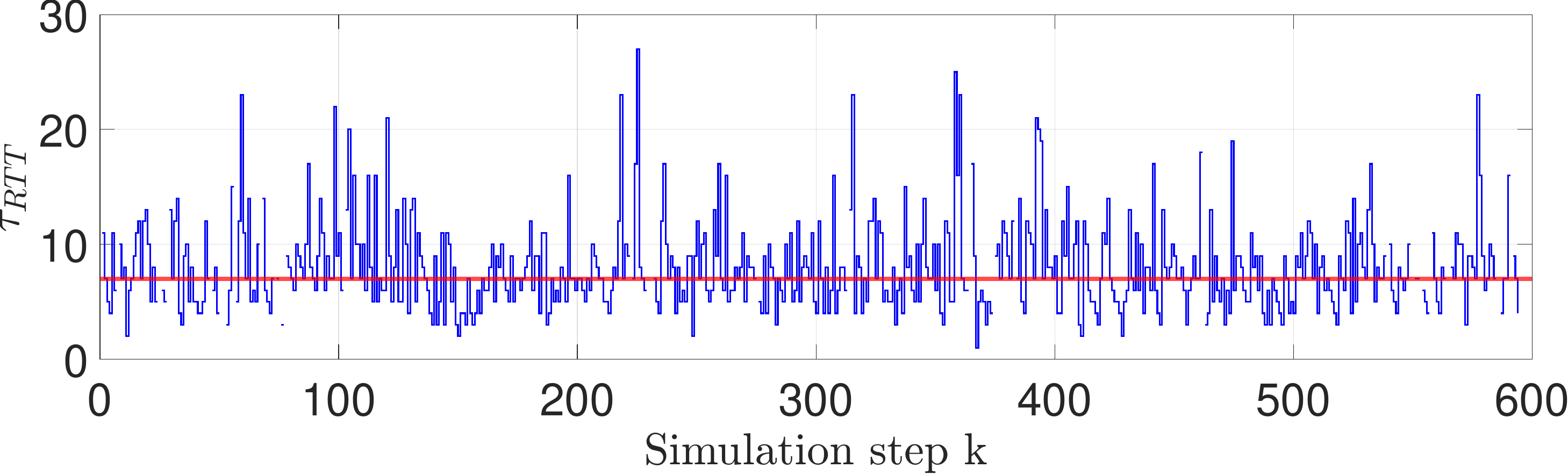}
\caption{Roundtrip time of the balance robot example. The red line marks the assumed roundtriptime $\bar{\tau}_{RTT}$. Everything above is considered a dropout.}
\label{fig:RTTBalanceRob}
\end{figure}
Due to its long sampling time, this simulation example is also suitable to investigate the timing of our proposed delay and dropout compensation method. In fig. \ref{fig:NWBehaviorCSTR} the network behavior under the proposed method is presented. The red line on the top represents the plant side, while the blue on the bottom represents the remote controller. Solid arrows are successfully transmitted packets. The dashed arrows either represent dropouts (red) or discarded messages (blue). The latter happens particularly often when packets are disordered. The red crosses on the bottom denote that the remote controller is in recovery mode. Above the plant and remote controller, several symbols show when the packet based on the measurement with the same symbol is applied at the actuator. As is evident, each packet is earliest applied after $\bar{\tau}_{RTT}$ time steps.

\begin{figure}[b]
\centering
\includegraphics[width=0.4\textwidth]{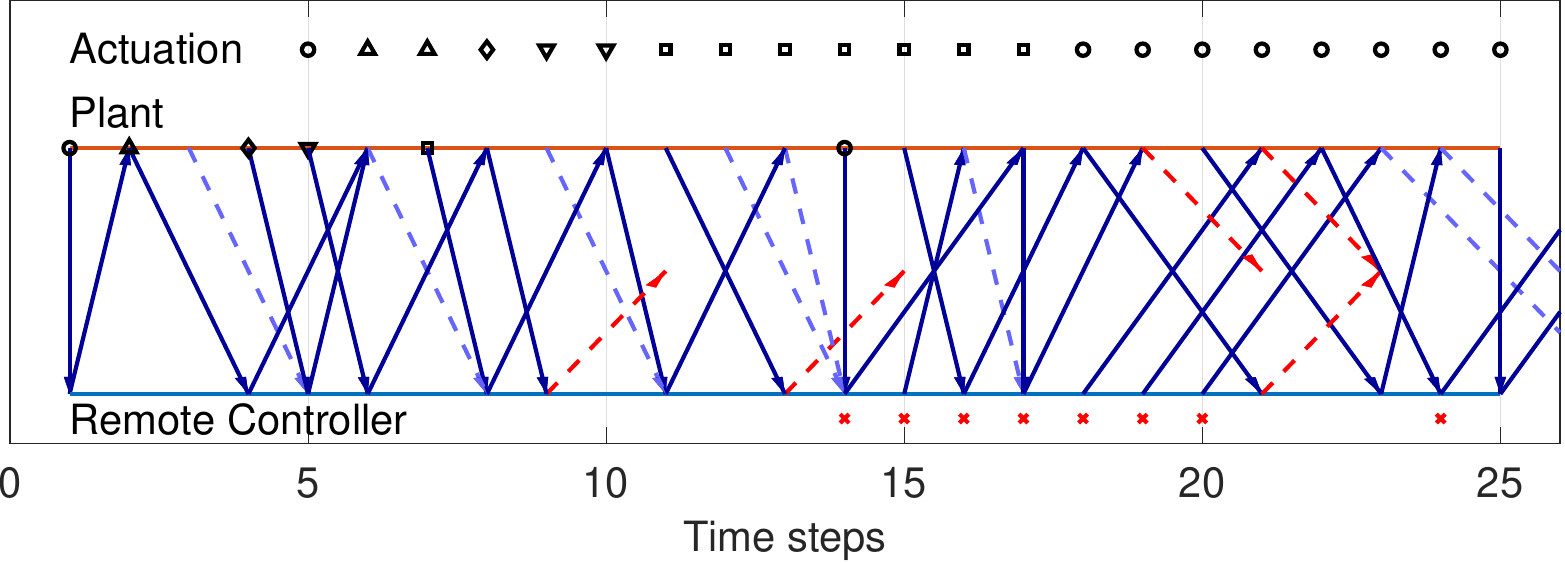}
\caption{Network behavior of the CSTR simulation. Dashed lines show dropped (red) and discarded (blue) packets. The symbols denote, which packet was used for the corresponding actuation. As expected the distance from the appearance of symbol on the plant line to its time of application on the actuation line is at least $N=6$.}
\label{fig:NWBehaviorCSTR}
\end{figure}
\section{Discussion}
%%% PLot Position CSTR
\begin{figure}[t]
\centering
\includegraphics[width=.5\textwidth]{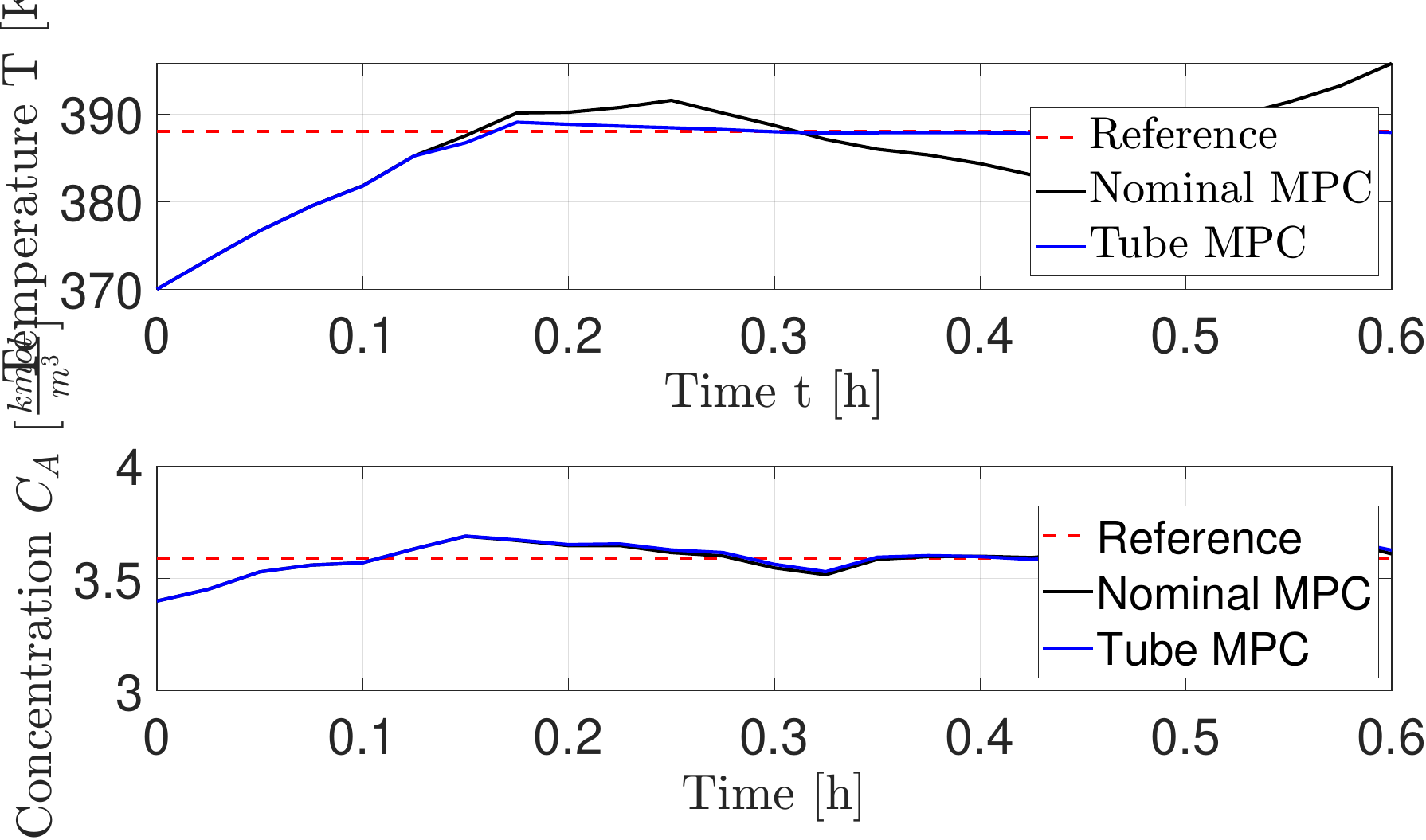}
\caption{State development over time for the CSTR simulation. The top depicts the temperature $T$, while the concentration $C_A$ is shown on the bottom. The tube approach handles the disturbances much better, particularly regarding temperature.}
\label{fig:PositionsCSTR}
\end{figure}
Our prediction consistent method stands out in two ways. First of all, it is suitable for a broad variety of communication networks as it does not rely on any form of acknowledged messaging as in \cite{Pin.2009b,L.Grune.2009} or time synchronization \cite{Varutti.2009a}. Additionally, our method performs consistently in networks, where average delays from controller to actuator are above the actual sample rate. As we do not consider a changing time horizon depending on the current delay, our method is also simpler to implement than, e.g., \cite{G.Pin.2021} while showing comparable performance, as can be seen from the similar simulations on the CSTR.
Secondly, it adds robustness to the networked control system using tubes, as was demonstrated in the simulations.
The main drawback of our approach is its rigidity due to buffering up to a bounded delay. This introduces artificial delays in situation, where the roundtrip was actually shorter, and thus it may decrease performance. Adapting online to a current best guess of the roundtrip delay or introducing time synchronization can leverage this disadvantage\\
The result from theorem \ref{thm:Theorem2} on the minimum invariant disturbance set states the expected. If $N < L$ we have to tighten the constraints more than for the nominal system due to delays and packet dropouts. However, the classical approach of approximating the tube with the maximum invariant disturbance set $\mathbb{S}_\infty$ still holds. Additionally, if the horizon is long enough, it dominates the influences of delays and dropouts and thus provides the same bounds as in the nominal scenario. Therefore, the condition seems rather mild, while the method ensures robust stability. 
%Plot Phase Portrait CSTR
\begin{figure}[t]
\centering
\includegraphics[width=.33\textwidth]{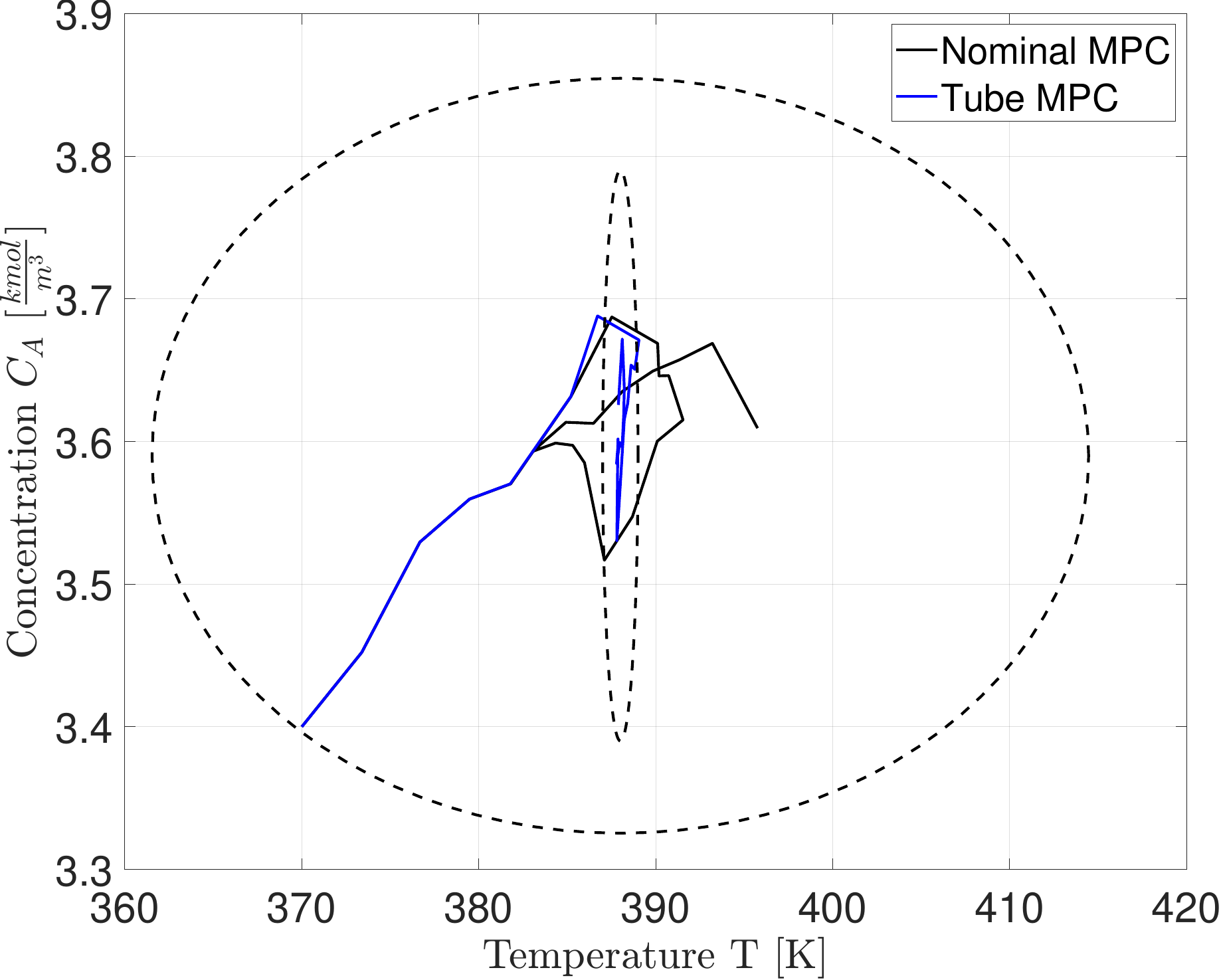}
\caption{Phase portrait of the CSTR simulation. While the nominal MPC approach leaves the target set, the tube MPC manages to stay inside even under disturbances.}
\label{fig:PhasePortraitCSTR}
\end{figure}
\section{Conclusion}
In the presented work we have introduced a novel predictive method for NCS that deals efficiently with communication constraints such as delays and packet losses. The technique relies on few, mild assumptions and is therefore applicable in a general setting. Thanks to the property of prediction consistency, it can be used with any type of provably stabilizing predictive control, linear or nonlinear. \\
Additionally, we have introduced tube MPC to the problem of Networked Predictive Control. We derived a bound on the minimum robust invariant disturbance set. This bound arises due to additional delays and longer error developments compared to the nominal setting due to the communication constraints. However, it does not restrict the method much, as in practice often a maximum positive invariant disturbance set is approximated, which contains the bounded set. \\
Further research may include investigations of adaptive techniques to the changing quality-of-service parameters of networks. 
Other areas of interest surround using online updated versions of tube MPC or learning MPC with the networked approach.

%%%%%%%%%%%%%%%%%%%%%%%%%%%%%%%%%%%%%%%%%%%%%%%%%%%%%%%%%%%%%%%%%%%%%%%%%%%%%%%%
\section*{APPENDIX}

\subsection{Proof of Lemma \ref{thm:Lemma1}}
\begin{proof}
    Prediction consistency is valid when all inputs that are newly applied at the plant were predicted with previously applied inputs \cite{L.Grune.2009}. Therefore, we need to ensure two things:
    \begin{enumerate}
        \item Buffer $B_c$ is consistent with $B_p$ before each state prediction at the remote controller
        \item The local controller at the plant side distinguishes wrongly predicted inputs and discards them
    \end{enumerate} 
    The second condition is ensured through the prediction consistency check (\ref{eq:PredictionConsistencyPlant}) (and subsequent pruning if applicable), which relies on the IDs of the input sequences. Alas, we need to look at the ID assignment at the remote controller to guarantee the first condition. \\
    The proof consists of analyzing the situations of packet losses and maximum delays in the backward (sensor-to-controller) and forward (controller-to-actuator) channels, both separately and together. Additionally, we need to differentiate between nominal and recovery mode.\\
   Delays and dropouts in the backward channel do not cause prediction inconsistencies. As long as a measurement arrives in time, the remote controller uses it to predict the next input series and sends the result to the plant. If a measurement packet $P_p$ is lost, the remote controller does not change behavior, as it will either do nothing when the controller is in nominal mode or it will rely on the previous measurement in recovery mode. On both occasions, the buffer contents of controller $B_c$ and plant $B_p$ stay consistent, and thus, the future state predictions are also consistent.\\
    The critical behavior occurs in the forward channel. Input delays are not problematic as long as the total roundtrip time stays below its upper-bound $\bar{\tau}_{RTT}$. Under these circumstances, the computed inputs arrive before or just at their application time, and the plant holds them in its buffer $B_p$ until they become valid. 
    If a drop occurs in the forward channel, the buffers on the plant and controller side differ from each other. As soon as the next measurement reaches the remote controller, the prediction inconsistency is detected as the test (\ref{eq:PredictionConsistentIDs}) fails, and the recovery mode is activated. Through the pruning strategy (\ref{eq:PruningOnError}) buffer $B_c$ is made consistent with $B_p$ again. By design, all predictions in recovery mode are solely based on $i_{p,last}$ of the last arrived measurement, which is the ID of the last consistently applied input at the plant. This guarantees that the next input trajectory, which arrives at the plant, is consistently predicted. In recovery mode, a new sequence is sent at every time step, regardless of a new measurement. The maximum number of time steps from the point of detection of an error at the plant until a new consistently predicted correction trajectory arrives is
    \begin{align}\label{eq:MaxTimeError}
        M = \bar{n}_l + \bar{\tau}_{RTT}.
    \end{align}
     This safety procedure is executed until a new measurement arrives at the remote controller, which carries the ID of a correction trajectory. Thus, the error correction is acknowledged. As a result, the remote controller needs to correct its buffer $B_c$ through (\ref{eq:PruneBufferOnCorrection}). This ensures, that Buffers $B_c$ and $B_p$ are consistent again. Now the remote controller can resume its nominal operation. 
    Ideally, if no losses occur, the recovery mode is resolved within one cycle of the maximum roundtrip time. In the worst case, however, the system must endure the maximum amount of steps as in (\ref{eq:MaxTimeError}). Additionally, if the system just recovered from a previous error, but the first value in nominal mode is dropped, it takes a full roundtrip time until the new error is detected by the remote controller. Therefore, we need to have at least 
    \begin{align}
    \label{eq:MaximumHorizon}
        N \geq \bar{n}_l+2\bar{\tau}_{RTT}
    \end{align}
    predicted input values to cope with the worst case situation. This poses a lower bound on the prediction horizon for MPC and the proof is concluded.
\end{proof}

\newpage
%%%%%%%%%%%%%%%%%%%%%%%%%%%%%%%%%%%%%%%%%%%%%%%%%%%%%%%%%%%%%%%%%%%%%%%%%%%%%%%%
\addtolength{\textheight}{-1cm}   % This command serves to balance the column lengths
                                  % on the last page of the document manually. It shortens
                                  % the textheight of the last page by a suitable amount.
                                  % This command does not take effect until the next page
                                  % so it should come on the page before the last. Make
                                  % sure that you do not shorten the textheight too much.
%\begin{thebibliography}{99}
\printbibliography
%\end{thebibliography}
%\bibliographystyle{IEEEtran}
%\bibliography{IEEEabrv,References}

\end{document}